\documentclass[11pt]{article}

\usepackage{fullpage}

\usepackage[
backend=bibtex,
maxbibnames=99,
style=numeric
]{biblatex}

\usepackage{tikz}
\usetikzlibrary{arrows,shapes,calc,positioning,decorations.pathreplacing,patterns,math}

\tikzset{
arc/.style={->,>=stealth',shorten >=1pt,thick}
}

\usepackage[english]{babel}
\usepackage{amsmath,amsthm,amssymb}
\usepackage{enumerate}
\usepackage{hyperref}

\newtheorem{theorem}{Theorem}
\newtheorem{lemma}[theorem]{Lemma}
\newtheorem{corollary}[theorem]{Corollary}
\newtheorem{observation}[theorem]{Observation}
\newtheorem*{theorem*}{Theorem}

\newcommand{\abs}[1]{\left| #1 \right|}

\bibliography{ORprecMakespan.bib}

\begin{document}

\title{Makespan Minimization with OR-Precedence Constraints}

\author{Felix Happach
    \thanks{
    Department of Mathematics and School of Management,
    Technische Universit\"{a}t M\"{u}nchen, Germany.
    This work has been supported by the Alexander von Humboldt Foundation with funds from the German Federal Ministry of Education and Research (BMBF).  Email address: {\tt felix.happach@tum.de} 
    }
}

\date{}

\maketitle

\begin{abstract}
We consider a variant of the NP-hard problem of assigning jobs to machines to minimize the completion time of the last job.
Usually, precedence constraints are given by a partial order on the set of jobs, and each job requires all its predecessors to be completed before it can start.
In his seminal paper, Graham (1966) presented a simple 2-approximation algorithm, and, more than 40 years later, Svensson (2010) proved that 2 is essentially the best approximation ratio one can hope for in general.

In this paper, we consider a different type of precedence relation that has not been discussed as extensively and is called OR-precedence.
In order for a job to start, we require that \emph{at least one} of its predecessors is completed -- in contrast to \emph{all} its predecessors.
Additionally, we assume that each job has a release date before which it must not start.

We prove that Graham's algorithm has an approximation guarantee of 2 also in this setting, and present a polynomial-time algorithm that solves the problem to optimality, if preemptions are allowed.
The latter result is in contrast to classical precedence constraints, for which Ullman (1975) showed that the preemptive variant is already NP-hard.
Our algorithm generalizes a result of Johannes (2005) who gave a polynomial-time algorithm for unit processing time jobs subject to OR-precedence constraints, but without release dates.
The performance guarantees presented here match the best-known ones for special cases where classical precedence constraints and OR-precedence constraints coincide.
\end{abstract}

\section{Introduction}\label{sec:introduction}

In this paper, we consider the problem of scheduling jobs with OR-precedence constraints on uniform parallel machines to minimize the total length of the project.
Let $[n] := \{1,\dots,n\}$ be the set of jobs and $m$ be the number of machines.
Each job $j \in [n]$ is associated with a processing time $p_j \geq 0$ and a release dates $r_j \geq 0$.
The precedence constraints are given by a directed graph $G = ([n],E)$.
The set of \emph{predecessors} of a job $j \in [n]$ is $\mathcal{P}(j) = \{i \in [n] \, | \, (i,j) \in E\}$.

A \emph{schedule} is an assignment of the jobs in $[n]$ to the machines such that (i) each job $j$ is processed by a machine for $p_j$ units of time, and (ii) each machine processes only one job at a time.
Depending on the problem definition, jobs may be allowed to preempt and continue on a different machine (\emph{preemptive scheduling}) or not (\emph{non-preemptive scheduling}).
The \emph{start time} and \emph{completion time} of job $j \in [n]$ are denoted by $S_j$ and $C_j$, respectively.
Note that $C_j \geq S_j + p_j$ and equality holds if job $j \in [n]$ is not preempted.

A schedule is called \emph{feasible}, if (i) $S_j \geq \min\{C_i \, | \, i \in \mathcal{P}(j)\}$, and (ii) $S_j \geq r_j$ for all jobs $j \in [n]$.
A job without predecessors may start at any point in time $t \geq r_j$.
In other words, every job with predecessors requires that \emph{at least one} of its predecessors is completed before it can start, and no job may start before it gets released.
A job $j$ is called \emph{available} at time $t \geq 0$, if $t \geq r_j$ and, unless $\mathcal{P}(j) = \emptyset$, there is $i \in \mathcal{P}(j)$ with $C_i \leq t$.
Our goal is to determine a feasible schedule that minimizes the \emph{makespan}, which is defined as $C_{\max} := \max_{j \in [n]} C_j$.
In an extension of the notation in~\cite{Johannes2005} and the three-field notation of Graham et al.~\cite{GrahamLawlerLenstraKan1979}, the preemptive and non-preemptive variant of this problem are denoted by $P \, | \, r_j, \, or\text{-}prec, \, pmtn \, | \, C_{\max}$ and $P \, | \, r_j, \, or\text{-}prec \, | \, C_{\max}$, respectively.

From now on we assume w.l.o.g.~that all processing times and release dates of jobs in $[n]$ are positive and non-negative integers, respectively.
Note that this can be done by suitable scaling and that any job with zero processing time may be disregarded.
As discussed below, the non-preemptive problem is NP-hard, which is why we are interested in approximation algorithms.
Let $\Pi$ be a minimization problem, and $\rho \geq 1$.
Recall that a $\rho$-approximation algorithm for $\Pi$ is a polynomial-time algorithm that returns a feasible solution with objective value at most $\rho$ times the optimal objective value.

\paragraph*{Non-Preemptive Scheduling.}
Garey and Johnson~\cite{GareyJohnson1978} proved that the non-preemptive variant is already strongly NP-hard in the absence of precedence constraints and release dates.
It remains NP-hard, even if the number of machines is fixed to $m=2$~\cite{LenstraKanBrucker1977}.
In his seminal paper, Graham~\cite{Graham1966} showed that a simple algorithm called \emph{List Scheduling} achieves an approximation guarantee of $2$:

\begin{center}
\emph{Consider the jobs in arbitrary order. Whenever a machine is idle, execute the next available job in the order on this machine. If there is no available job, then wait until a job completes.}
\end{center}

If the jobs are sorted in order of non-increasing processing times, then List Scheduling is a $\frac{4}{3}$-approximation~\cite{Graham1969}.
Hochbaum and Shmoys~\cite{HochbaumShmoys1988} presented a $(1+\varepsilon)$-approximation for $P \, | \, | \, C_{\max}$, which was improved in running time to the currently best-known by Jansen~\cite{Jansen2010}.
Mnich and Wiese~\cite{MnichWiese2015} showed that $P \, | \, | \, C_{\max}$ is fixed parameter tractable with parameter $\max_{j \in [n]} p_j$.
If we add non-trivial release dates, then List Scheduling with an arbitrary job order is a 2-approximation~\cite{HallShmoys1989}, and it is $\frac{3}{2}$-approximate if the jobs are sorted in order of non-increasing processing times~\cite{ChenVestjens1997}.
Hall and Shmoys~\cite{HallShmoys1989} provided a $(1+\varepsilon)$-approximation for $P \, | \, r_j \, | \, C_{\max}$.

In contrast to OR-precedence constraints that are considered in this paper, the standard precedence constraints, where each job requires that \emph{all} its predecessors are completed, will be called \emph{AND-precedence constraints}.
Minimizing the makespan with AND-precedence constraints is strongly NP-hard, even if the number of machines is fixed to $m = 2$ and the precedence graph consists of disjoint paths~\cite{DuLeungYoung1991}.
List Scheduling is still 2-approximate in the presence of AND-precedence constraints if the order of the jobs is consistent with the precedence constraints~\cite{Graham1966,Graham1969}.
The approximation factor can also be preserved for non-trivial release dates~\cite{HallShmoys1989}.
Assuming a variant of the Unique Games Conjecture~\cite{Khot2002} together with a result of Bansal and Khot~\cite{BansalKhot2009}, Svensson~\cite{Svensson2010} proved that this is essentially best possible.

If the precedence constraints are of AND/OR-structure and the precedence graph is acyclic, then the problem without release dates still admits a $2$-approximation algorithm~\cite{GilliesLiu1995}.
Erlebach, K\"a\"ab and M\"ohring~\cite{ErlebachKaabMohring2003} showed that the assumption on the precedence graph is not necessary.
Both results first transform the instance to an AND-precedence constrained instance by fixing a predecessor of the OR-precedence constraints.
Then they solve the resulting instance with AND-precedence constraints using List Scheduling.
Our first result shows that the makespan of \emph{every} feasible schedule without unnecessary idle time on the machines is at most twice the optimal makespan, even if non-trivial release dates are involved.

\begin{theorem}\label{thm:2approximation}
List Scheduling is a $\left(2 - \frac{1}{m} \right)$-approximation for $P \, | \, r_j, \, or\text{-}prec \, | \, C_{\max}$.
\end{theorem}

The proof of Theorem~\ref{thm:2approximation} is contained in Section~\ref{sec:nonpreemptive}.
The key ingredient for proving the performance guarantee is the concept of \emph{minimal chains} that we introduce in Section~\ref{sec:preliminaries}.
Informally the length of the minimal chain of job $j \in [n]$ is the amount of extra time we need to complete $j$.
The minimal chain of $j$ is the set of jobs in $[n] \setminus S$ that have to be processed in order to complete $j$ in that time.

\paragraph*{Preemptive Scheduling.}
If preemptions are allowed the algorithm of McNaughton~\cite{McNaughton1959} computes an optimal schedule in the absence of release dates and precedence constraints.
Ullman~\cite{Ullman1975} showed that the problem with AND-precedence constraints is NP-hard, even if all jobs have unit processing time.
Note that if $p_j = 1$ for all jobs $j$, then there is no benefit in preemption.
This implies that the preemptive problem with AND-precedence constraints is also NP-hard.
However, the preemptive variant becomes solvable in polynomial time for certain restricted precedence graphs.
Precedence graphs that consist of outtrees are of special interest to us, since then AND- and OR-precedence constraints coincide.

A number of polynomial-time algorithms were proposed for AND-precedence constraints in form of an outtree.
Hu~\cite{Hu1961} proposed the first such algorithm for unit processing time jobs, and Brucker, Garey and Johnson~\cite{BruckerGareyJohnson1977} presented an algorithm that can also deal with non-trivial release dates.
Muntz and Coffman~\cite{MuntzCoffman1970} gave a polynomial-time algorithm, if preemptions are allowed.
The algorithm of Gonzalez and Johnson~\cite{GonzalezJohnson1980} has an asymptotically better running time and uses fewer preemptions than the one in~\cite{MuntzCoffman1970}.
Finally Lawler~\cite{Lawler1982} proposed a polynomial-time algorithm for the preemptive variant that can deal with non-trivial release dates, if the precedence graph consists of outtrees.\footnote{Note that Lawler's algorithm~\cite{Lawler1982} generalizes those of~\cite{Hu1961,MuntzCoffman1970,BruckerGareyJohnson1977,GonzalezJohnson1980}.}

For general OR-precedence constrained unit processing time jobs, Johannes~\cite{Johannes2005} presented a polynomial-time algorithm that is similar to Hu's algorithm~\cite{Hu1961}.
We improve on this result by analyzing the structure of an optimal solution of $P \, | \, r_j, \, or\text{-}prec, \, pmtn \, | \, C_{\max}$.
More precisly, we show that there is an optimal preemptive schedule where each job is preceded by its minimal chain.
We then exploit this structure to transform the instance into an equivalent AND-precedence constrained instance, where we can apply known algorithms of e.g.~\cite{Hu1961,MuntzCoffman1970,BruckerGareyJohnson1977,GonzalezJohnson1980,Lawler1982}.
Thereby we obtain our second result. The proof is contained in Section~\ref{sec:preemption}.

\begin{theorem}\label{thm:preemption}
$P \, | \, r_j, \, or\text{-}prec, \, pmtn \, | \, C_{\max}$ can be solved to optimality in polynomial time.
\end{theorem}

Since there is no need to preempt if $p_j = 1$ for all $j \in [n]$, we immediately obtain the following corollary. This generalizes the aforementioned result of~\cite{Johannes2005}.

\begin{corollary}\label{cor:unitprocessingtime}
$P \, | \, r_j, \, or\text{-}prec, \, p_j = 1 \, | \, C_{\max}$ can be solved to optimality in polynomial time.
\end{corollary}


\section{Preliminaries and Minimal Chains}\label{sec:preliminaries}

In order to simplify some arguments, we introduce a dummy job $s$ with $p_s = r_s = 0$ that shall precede all jobs.
That is, we assume that the set of jobs is $N = [n] \cup \{s\}$, and introduce an arc $(s,j)$ for all $j \in [n]$ with $\mathcal{P}(j) = \emptyset$ in the precedence graph $G$.
Note that there is a feasible schedule, if and only if every job $j \in [n]$ is reachable from $s$ in $G = (N,E)$. 
In particular, we can decide in linear time, e.g.~via breadth-first-search, whether there exists a feasible schedule.
Henceforth, we will assume that the instances we consider admit a feasible schedule.

Note that $P \, | \, or\text{-}prec \, | \, C_{\max}$ is a generalization of $P \, | \, | \, C_{\max}$ which is already strongly NP-hard~\cite{GareyJohnson1978}.
If $G$ is an outtree rooted at $s$, then OR- and AND-precedence constraints are equivalent.
The NP-hardness result of Du, Leung and Young~\cite{DuLeungYoung1991} implies that the problem remains strongly NP-hard, even if the number of machines is fixed.

\begin{observation}
$Pm \, | \, or\text{-}prec \, | \, C_{\max}$ is strongly NP-hard for all $m \geq 2$.
\end{observation}

In order to analyze the performance of our algorithms, we use the concept of so-called \emph{minimal chains}.
Informally, a minimal chain of a job $k$ is a set of jobs that need to be scheduled so that $k$ can complete as early as possible.
To define minimal chains properly, we use the notion of an \emph{earliest start schedule}~\cite{ErlebachKaabMohring2003,MohringSkutellaStork2004,Johannes2005}.
Although these schedules are well-defined for general AND/OR-scheduling, we only need and define them in the OR-scheduling context.

The \emph{earliest start schedule} is defined as a schedule on an infinite number of machines such that $(i)$ a job $j$ without predecessors starts at time $r_j$, and $(ii)$ a job $j$ with $\mathcal{P}(j) \not= \emptyset$ starts at time $\max\{r_j,\min\{C_i \, | \, i \in \mathcal{P}(j)\}\}$.
Clearly, an earliest start schedule respects the OR-precedence constraints of the instance, since every job is preceded by at least one of its predecessors according to $(ii)$.
Also, the completion time of a job in any feasible schedule on $m$ machines is bounded from below by its completion time in the earliest start schedule.
That is, if $C_j$ denotes the completion time of job $j$ in the earliest start schedule, the optimum makespan satisfies $C^*_{\max} \geq \max\{C_j \,| \, j \in N\}$.
Note that an earliest start schedule is not necessarily unique, but the start and completion times of all jobs are fix.
Earliest start schedules can be constructed in polynomial time by iteratively scheduling every job as early as possible~\cite{ErlebachKaabMohring2003}.

Let $k \in N$ and let $C_j$ be the completion time of $j \in N$ in the earliest start schedule.
A set $L \subseteq N$ is called \emph{minimal chain of $k \in N$} if $L \in \mathcal{S}$ is inclusion-minimal such that $k \in L$ and $\max_{j \in L} C_j = C_k$.
The set of minimal chains of $k$ is denoted by $\mathcal{MC}(k)$, and the \emph{length of the minimal chain of $k$} is $mc(k) := C_k$.

We can construct a minimal chain of $k$ by iteratively tracing back predecessors that delay job~$k$ in the earliest start schedule.
That is, starting at $k$, we mark one of its predecessors $j$ with $C_j = S_k$, and then proceed with $j$ in the same manner, i.e., we mark a predecessor $i$ of $j$ with $C_i = S_j$, and so on, until we reach a job $i'$ that starts at its release date.
If $i'$ has no predecessors, we are done.
If $\mathcal{P}(i') \not= \emptyset$, we mark a predecessor $j'$ of $i'$ with $C_{j'} \leq S_{i'}$, and continue with $j'$ as described above.
The marked jobs now correspond to a minimal chain of~$k$.
That is, a minimal chain $L = \{j_1,\dots,j_\ell\} \in \mathcal{MC}(k)$ is a path in $G$ with $\mathcal{P}(j_1) = \emptyset$, $j_q \in \mathcal{P}(j_{q+1})$ for all $q \in [\ell -1]$ and $j_\ell = k$ such that $S_{j_1} = r_{j_1}$ and $S_{j_q} = \max\{r_{j_q},C_{j_{q-1}}\}$ for all $2 \leq q \leq \ell$.
We call $j_q$ \emph{the predecessor of $j_{q+1}$ in $L$} for $q \in [\ell -1]$ and denote this by $\mathcal{P}_L(j_{q+1}) := \{j_q\}$.
A job $j_h \in L$ is said to \emph{dominate} the minimal chain $L$ if $mc(k) = r_{j_h} + \sum_{q = h}^\ell p_{j_q}$.

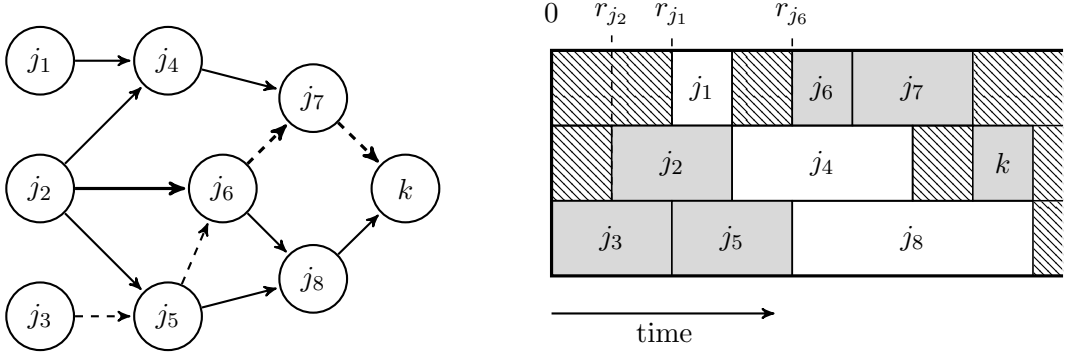
\begin{figure}
\centering
\begin{tikzpicture}[->,>=stealth',shorten >=1pt,auto,node distance=1.7cm,thick]
		\node[draw,thick,circle,minimum size=9mm] (1) {$j_1$};
		\node[draw,thick,circle,minimum size=9mm] (2) [below of=1] {$j_2$};
		\node[draw,thick,circle,minimum size=9mm] (3) [below of =2] {$j_3$};
		\node[draw,thick,circle,minimum size=9mm] (4) [right of=1] {$j_4$};
		\node[draw,thick,circle,minimum size=9mm] (5) [right of=3] {$j_5$};
		\node[draw,thick,circle,minimum size=9mm] (6) [right=1.5cm of 2] {$j_6$};
		\node[draw,thick,circle,minimum size=9mm] (7) [above right of=6] {$j_7$};
		\node[draw,thick,circle,minimum size=9mm] (8) [below right of=6] {$j_8$};
		\node[draw,thick,circle,minimum size=9mm] (k) [right=1.5cm of 6] {$k$};

		\path (1) edge (4)
				(4) edge (7)
				(7) edge[dashed,very thick] (k);
		\path (2) edge[very thick] (6)
				(6) edge (8)
				(8) edge (k);
		\path (2) edge (5)
				(5) edge[dashed] (6)
				(6) edge[dashed,very thick] (7);
		\path (2) edge (4);
		\path (3) edge[dashed] (5)
				(5) edge (8);
						
\end{tikzpicture}
\hspace*{1cm}
\begin{tikzpicture}[semithick]
			\tikzmath{\s = 0.8;} 
			
			\coordinate (l1) at (0,2);
			\coordinate (l2) at (0,1);
			\coordinate (l3) at (0,0);
			\coordinate (l4) at (0,-1);
			
			\coordinate (r1) at (8.5*\s,2);
			\coordinate (r2) at (8.5*\s,1);
			\coordinate (r3) at (8.5*\s,0);
			\coordinate (r4) at (8.5*\s,-1);
			
			\draw[pattern=north west lines] (l1) rectangle ($(l2) + (2*\s,0)$);
			\draw ($(l1) + (2*\s,0)$) rectangle ($(l2) + (3*\s,0)$) node[midway] {$j_1$};
			
			\draw[pattern=north west lines] ($(l1) + (3*\s,0)$) rectangle ($(l2) + (4*\s,0)$);
			\draw[fill={gray!30}] ($(l1) + (4*\s,0)$) rectangle ($(l2) + (5*\s,0)$) node[midway] {$j_6$};				

			\draw[white,pattern=north west lines] ($(l1) + (7*\s,0)$) rectangle (r2);
			\draw[fill={gray!30}] ($(l1) + (5*\s,0)$) rectangle ($(l2) + (7*\s,0)$) node[midway] {$j_7$};		
			
			\draw[pattern=north west lines] (l2) rectangle ($(l3) + (1*\s,0)$);
			\draw[fill=gray!30] ($(l2) + (1*\s,0)$) rectangle ($(l3) + (3*\s,0)$) node[midway] {$j_2$};
			
			\draw ($(l2) + (3*\s,0)$) rectangle ($(l3) + (6*\s,0)$) node[midway] {$j_4$};			
			\draw[pattern=north west lines] ($(l2) + (6*\s,0)$) rectangle ($(l3) + (7*\s,0)$);
				
			\draw[white,pattern=north west lines] ($(l2) + (8*\s,0)$) rectangle (r3);
			\draw[fill={gray!30}] ($(l2) + (7*\s,0)$) rectangle ($(l3) + (8*\s,0)$) node[midway] {$k$};			
			
			\draw[fill=gray!30] (l3) rectangle ($(l4) + (2*\s,0)$) node[midway] {$j_3$};
			\draw[fill=gray!30] ($(l3) + (2*\s,0)$) rectangle ($(l4) + (4*\s,0)$) node[midway] {$j_5$};
			
			\draw[white,pattern=north west lines] ($(l3) + (8*\s,0)$) rectangle (r4);
			\draw ($(l3) + (4*\s,0)$) rectangle ($(l4) + (8*\s,0)$) node[midway] {$j_8$};		
						
			\draw[very thick] (r1) -- (l1) -- (l4) -- (r4);
			\draw (l2) -- (r2);
			\draw (l3) -- (r3);
			
			\node at ($(l1) + (0,0.5)$) {$0$};
			\draw[dashed] ($(l2) + (2*\s,0)$) -- ($(l1) + (2*\s,0.2)$) node[above] {$r_{j_1}$};		
			\draw[dashed] ($(l3) + (1*\s,0)$) -- ($(l1) + (1*\s,0.2)$) node[above] {$r_{j_2}$};	
			\draw[dashed] ($(l2) + (4*\s,0)$) -- ($(l1) + (4*\s,0.2)$) node[above] {$r_{j_6}$};	
			
			\draw[arc] (0,-1.5) -- (3,-1.5) node[midway,yshift=-7pt] {time};
			
\end{tikzpicture}
\caption{{\small An instance on nine jobs with processing times $p_{j_1}  =p_{j_6} = p_k = 1$, $p_{j_2} = p_{j_3} = p_{j_5} = p_{j_7} = 2$, $p_{j_4} = 3$, $p_{j_8} = 4$ and release dates $r_{j_1} = 2$, $r_{j_2} = 1$, $r_{j_6} = 4$, $r_j = 0$ for all other jobs~$j$ (left) and an earliest start schedule (right). The set of minimal chains of $k$ is $\mathcal{MC}(k)= \{\{j_2,j_6,j_7,k\},\{j_3,j_5,j_6,j_7,k\}\}$ with $mc(k) = 8$. The chain $\{j_2,j_6,j_7,k\}$ is dominated by $j_6$, and $\{j_3,j_5,j_6,j_7,k\}$ is dominated by jobs $j_3$ and $j_6$. The paths in $G$ that correspond to the minimal chains in $\mathcal{MC}(k)$ are depicted dashed and thick, respectively. Jobs in minimal chains are highlighted in gray.}}
\label{fig:makespan:example:minchain}
\end{figure}

In the following, we denote the completion times in an optimal schedule by $C^*_j$ (for $j \in [n]$) and its makespan by $C^*_{\max}$.
Also, we will sometimes denote an optimal schedule by $C^*$ and the schedule with completion times $C_j$ (for $j \in [n]$) by $C$.
There are two trivial lower bounds on the optimal makespan.
First, any feasible schedule cannot do better than splitting the total processing load equally among all machines, so $C^*_{\max} \geq \frac{1}{m} \sum_{j \in N} p_j$.
Second, every job requires at least one of its predecessors to be completed before it can start.
If we start with an empty schedule, the earliest completion time of job $j$ is by definition equal to the length of its minimal chain w.r.t.~the empty set.
Thus, $C^*_{\max} \geq \max_{j \in N} \ mc(j)$.


\section{List Scheduling Without Preemptions}\label{sec:nonpreemptive}

Erlebach et al.~\cite{ErlebachKaabMohring2003} presented a 2-approximation algorithm for minimizing the makespan with AND/OR-precedence constraints.
The algorithm transforms the instance to an AND-instance by fixing an OR-predecessor for each job, and then applies List Scheduling.
We show that List Scheduling without transforming the instance is already $2$-approximate for OR-precedence constraints, even with non-trivial release dates.
The proof is similar to~\cite{HallShmoys1989}.
Since we consider OR-precedence constraints, we need the notion of minimal chains to bound the amount of idle time on the machines.
For completeness, we restate Theorem~\ref{thm:2approximation} here.

\begin{theorem*}
List Scheduling is a $\left(2 - \frac{1}{m} \right)$-approximation for $P \, | \, r_j, \, or\text{-}prec \, | \, C_{\max}$.
\end{theorem*}

\begin{proof}
Consider the schedule returned by List Scheduling, and let $S_j$ and $C_j$ be the start and completion time of job $j \in [n]$.
Let $l \in [n]$ be the job that completes last, i.e.~$C_l = C_{\max}$.
Let $I  \subseteq [0,S_l]$ be the union of all time intervals $I_1, \dots, I_b$ where some machine is idle.
If $I = \emptyset$, then all machines are busy before time $S_l$ with jobs in $N \setminus \{l\}$. Hence
\begin{equation*}
C_{\max} = S_l + p_l \leq \frac{1}{m} \sum_{j \not= l} p_j + p_l = \frac{1}{m} \sum_{j \in N} p_j + \left(1- \frac{1}{m} \right) p_l \leq \left(2-\frac{1}{m} \right)C^*_{\max}.
\end{equation*}

So suppose there is idle time, and let $I$ be the union of all intervals in which some machine is idle.
Let $S \subseteq N$ be a set of jobs such that $S$ is a path in the precedence graph from the source $s$ to $k$.
At every point in time $t \in I$, a job in $S$ is either not yet released, or is currently running on some machine.
Otherwise, there is an unscheduled available job in $S$ that can be processed at time~$t$.

Let $L' \in \mathcal{MC}(k)$ be a minimal chain of $k$ and enumerate the jobs $L' = \{j_1,\dots,j_\ell\}$ such that $j_\ell = k$, $\mathcal{P}(j_1) = \emptyset$, and $j_q \in \mathcal{P}(j_{q+1})$ for all $q \in [\ell -1]$.
Recall that $L'$ is a path in $G$, so, at every idle point in time, some job of $L'$ is either being processed or not yet released.
That is, the total idle time is $\abs{I} \leq mc(k)$, but we can even get an even stronger bound.

Let $h \in [\ell]$ be maximal such that $j_h$ dominates the minimal chain $L'$, i.e., $mc(k) = r_{j_h} + \sum_{q = h}^\ell p_{j_q}$.
Let $L := \{j_h,\dots,j_\ell\}$, and consider the points in time $I_L := [0;r_{j_h}] \cup \bigcup_{j \in L} [S_j;C_j]$ when $j_h$ is not yet released or some job in $L$ is being processed.
Note that the intervals $[S_j;C_j]$ for $j \in L$ are not necessarily disjoint, because jobs in $L$ might run in parallel, i.e., $\abs{I_L} \leq r_{j_h} + \sum_{j \in L} p_j$.
W.l.o.g., we can assume that at least one machine is running during $[0;r_{j_h}]$.
Otherwise, if all machines were idle at some point $t \in [0;r_{j_h}]$, then all jobs $j$ with $r_j \leq t$ are already completed at time $t$.
Thus, also in the optimum solution, no machine is running at time $t$, so we can disregard these time slots where no machine is running at all.
That is, during $I_B := [0;C_{\max}] \setminus I_L$, all machines are busy with jobs in $N \setminus L$ and the total processing load of jobs that are running in $I_B$ is less or equal than $\sum_{j \notin L} p_j - r_{j_h}$.
Hence, $\abs{I_B} \leq \frac{1}{m} ( \sum_{j \notin L} p_j - r_{j_h})$ and we obtain

\begin{align*}\
C_{\max}  &= C_k = \abs{I_B} + \abs{I_L} \leq \frac{1}{m} \left( \sum_{j \notin L} p_j - r_{j_h} \right) + r_{j_h} + \sum_{q = h}^\ell p_{j_q} =\\
&=  \frac{1}{m} \sum_{j \in N} p_j + \left( 1-\frac{1}{m}\right) mc(k) \leq \left( 2-\frac{1}{m}\right) C^*_{\max} 
\end{align*}
This proves the claim.
\end{proof}

\begin{corollary}\label{cor:releasedates:singlemachine}
List Scheduling solves $1 \, | \, r_j, \, or\text{-}prec \, | \, C_{\max}$ to optimality.
\end{corollary}


\section{A Polynomial-Time Algorithm for Special Cases}\label{sec:preemption}

In this section, we consider the preemptive problem $P \, | \, r_j, \, or\text{-}prec, \, pmtn \, | \, C_{\max}$ and prove Theorem~\ref{thm:preemption}. 
Recall that all processing times and release dates of jobs in $[n]$ are positive and non-negative integers, respectively.
So preemptive and non-preemptive scheduling of unit processing time jobs are equivalent, since there is no need to preempt, which proves Corollary~\ref{cor:unitprocessingtime}.

In contrast to the non-preemptive instance, an optimal preemptive schedule will never have idle time, if there are available jobs.
Without preemption, it could make sense to wait for some job $j$ to finish (i.e.~have idle time), although there is an available job $k$.
The reason might be that we want to process a successor $i$ of  $j$ rightaway.
However, if we allow preemption, then we could just schedule a fraction of $k$, and once $j$ completes, we preempt $k$ and process $i$.

We first derive some necessary notation, and then present a polynomial-time algorithm that computes an optimal preemptive schedule. 
Fix $L_j \in \mathcal{MC}(j)$ for all $j \in N$.
The collection of minimal chains $\{L_j \, | \, j \in N\}$ is called \emph{closed}, if $i \in L_j$ implies $L_i \subseteq L_j$ for all $j \in N$.
Note that we can always choose $L_i \subseteq L_j$ for all $i \in L_j$, since (informally) subpaths of shortest paths are shortest paths.
Hence, if we compute minimal chains $L_1,\dots,L_n$ using the procedure described in Section~\ref{sec:preliminaries}, we may assume that $\{L_1,\dots,L_n\}$ is closed.
We say an arc $(i,j) \in E$ is \emph{in line with the minimal chain} $L_j$ if $i \in L_j$.
Recall that all processing times are strictly positive and $L_j \in \mathcal{MC}(j)$. So if $(i,j) \in E$ is in line with $L_j$ then $i \in \mathcal{P}(j)$.

Our algorithm, which we refer to as \textsc{AlgoPmtn}, works as follows.
First, compute a closed collection of minimal chains $\{L_j \, | \, j \in N\}$.
Then, transform the instance to an instance with AND-precedence constraints by deleting all arcs that are not in line with $L_1,\dots,L_n$.
(Note that the resulting graph $G'$ is an outtree.)
Now, apply a polynomial-time algorithm for the resulting AND-instance to compute an optimal preemptive schedule.
(Recall that we can compute optimal preemptive schedules for these special cases in polynomial time, see e.g.~\cite{Hu1961,MuntzCoffman1970,BruckerGareyJohnson1977,GonzalezJohnson1980,Lawler1982}.
We use the algorithm of Lawler~\cite{Lawler1982}, but instead, depending on the setting, we could also use any of the other algorithms.)

We prove that \textsc{AlgoPmtn} works correctly by analyzing the structure of an optimal preemptive schedule.
More precisly, we show that for any closed collection of minimal chains, there is an optimal preemptive schedule that is feasible for the transformed graph $G'$.
Before we are able to prove Theorem~\ref{thm:preemption}, we need some additional notation.

If jobs are allowed to preempt, we need to ``keep track'' how much of the minimal chain of a job is already processed at every point in time. 
To formalize this, we split every job $j \in [n]$ into $p_j$ jobs $j_1,\dots,j_{p_j}$ of unit processing time. 
The predecessors of these jobs are $\mathcal{P}(j_1) = \{i_{p_i} \ | \ (i,j) \in E \}$ and $\mathcal{P}(j_u) = \{j_{u-1}\}$ for all $2 \leq u \leq p_j$.
The release dates are $r_{j_u} = r_{j}$ for all $j \in N$ and $u \in [p_j]$.
As before, we add a dummy job $s$ with $p_s = r_s = 0$ and $\mathcal{P}(j_1) = \{s\}$ if $\mathcal{P}(j) = \emptyset$ for $j \in [n]$.
We refer to this instance as the \emph{preemtive instance} and denote the set of jobs by $N^{(p)}$.

Note that, if all jobs have unit processing time, then $N^{(p)} = N$.
We informally extend definition of $mc(k)$ to \emph{fractions of jobs} via the original definition on the preemptive instance.
Note that (the lengths of) all minimal chains coincide with the non-preemptive instance.
In particular, all lower bounds on the makespan are still valid, and $i \in L_j$ implies $i_{1},\dots,i_{p_i} \in L_{j_u}$ for all $u \in [p_j]$.
Since minimal chains in the non-preemptive and preemptive instance coincide, $\{L_j \, | \, j \in N^{(p)}\}$ is closed iff $\{L_j \, | \, j \in N\}$ is closed.
Two distinct jobs $i, j \in N^{(p)}$ are called \emph{inverted} w.r.t.~the closed collection of minimal chains $\{L_k \ | \ k \in N^{(p)}\}$ in the schedule $C$, if $i \in L_j$ and $C_i \geq C_j$.
Let $I_{C}$ be the number of inversions in the schedule $C$.

Lemma~\ref{lem:cmax:swap} describes a procedure that swaps two jobs $k,l \in N^{(p)}$ that are scheduled consecutively.
We will apply this procedure to show that there always exists an optimal solution without inversions (see Lemma~\ref{lem:cmax:preemption:structureofOPT}).
For the notation of Lemma~\ref{lem:cmax:swap}, we forget about release dates, i.e.~consider schedules for $P \, | \, or\text{-}prec, \, pmtn \, | \, C_{\max}$.
We describe how to incorporate release dates in the proof of Lemma~\ref{lem:cmax:preemption:structureofOPT}, which is the key lemma for the correctness of \textsc{AlgoPmtn}.

\begin{lemma}\label{lem:cmax:swap}
Let $\{L_j \ | \ j \in N^{(p)}\}$ be a closed collection of minimal chains and $C^*$ be a feasible preemptive schedule.
Let $i \in N^{(p)}$ with $C^*_i \geq 2$, and let $S_i = \{j \in N^{(p)} \, | \, C^*_{j} = C^*_i -1\}$ be the jobs scheduled directly before~$i$.
Assume that $\abs{S_i}= m$ and $C^*_j \leq C^*_i -2$ for $j \in \mathcal{P}_{L_i}(i)$.\footnote{So moving $i$ to $[C^*_i -2;C^*_i -1]$ does not violate its precedence constraints or cause an inversion.}
Then there is $k \in S_i$ such that swapping $i$ and $k$, i.e., setting $C'_{i} = C^*_{i} -1 = C^*_{k}$, $C'_{k} = C^*_{k} +1 = C^*_{i}$ and $C'_j = C^*_j$ for all $j \in N^{(p)} \setminus \{i,k\}$, yields a feasible schedule with $C'_{\max} = C^*_{\max}$ and $I_{C'} \leq I_{C^*}$.
\end{lemma}

\begin{proof}
To shorten notation, set $t := C^*_i -1$.	
Note that the makespan does not change if we swap two unit processing time jobs.
Let $J_i = \{j \in N^{(p)} \setminus \{i\}  \ | \ C^*_j = t+1\}$ be the jobs running in parallel to $i$ on the other machines.
Note that $\abs{J_i} \leq m-1$, and recall that there are $\abs{S_i} = m$ jobs that are being processed directly before $i$.
For $j \in N^{(p)}$, let $\mathcal{A}^*_j := \{j' \in N^{(p)} \, | \, C^*_{j'} < C^*_j\}$ be the set of jobs that complete before $j$ starts.

\begin{figure}
\centering
\begin{tikzpicture}[decoration={brace,amplitude=5pt,mirror}]
		\draw[very thick] (0,0) -- (4,0);
		
		\draw[semithick] (0,-1) -- (1,-1);
		\draw[semithick] (3,-1) -- (4,-1);
		
		\draw[semithick] (0,-2) -- (2,-2);
		\draw[semithick] (3,-2) -- (4,-2);
		
		\draw[semithick] (0,-3) -- (1,-3);
		\draw[semithick] (2,-3) -- (4,-3);
		
		\draw[semithick] (0,-4) -- (1,-4);
		\draw[semithick] (2,-4) -- (4,-4);		
		
		\draw[very thick] (0,-5) -- (4,-5);
		\draw (1,0) -- (1,-5);
		\draw (2,0) -- (2,-5);	
		\draw (3,0) -- (3,-5);
		
		\node at (1,0.5) {$t-1$};
		\node at (2,0.5) {$t$};
		\node at (3,0.5) {$t+1$};		
		
		\draw[gray,dashed] (1,-1) -- (3,-1);
		\draw[gray,dashed] (2,-2) -- (3,-2);
		\draw[gray,dashed] (1,-3) -- (2,-3);		
		\draw[gray,dashed] (1,-4) -- (2,-4);		
		
		\draw[very thick] (1,0) -- (2,0) -- (2,-2) -- (1,-2) -- (1,0);
		\node at (1.5,-1.2) {$S'$};		
		\draw[very thick] (2,0) -- (3,0) -- (3,-3) -- (2,-3) -- (2,0);
		\node at (2.5,-1.2) {$J'$};
		
		\draw[->, thick] (1.5,-0.5) -- (2.5,-0.5);
		\draw[->, thick] (1.5,-1.5) -- (2.5,-1.5);
		\draw[->, thick] (1.5,-1.5) -- (2.5,-2.5);		
				
		\draw[very thick] (1,-2) -- (2,-2) -- (2,-5) -- (1,-5) -- (1,-2);
		\node at (1.5,-2.7) {$S$};
		
		\node at (2.5,-4.5) {$i$};
		\node at (1.5,-3.5) {$k$};
		\draw[decorate,semithick] (1,-5) -- (2,-5) node[midway, yshift =-0.5cm] {$S_i$};
		\draw[decorate,semithick] (2,-5) -- (3,-5) node[midway, yshift =-0.5cm] {$J_i \cup \{i\}$};

		\draw[pattern=north west lines] (2,-3) -- (3,-3) -- (3,-4) -- (2,-4) -- (2,-3);
\end{tikzpicture}
\hspace*{2cm}
\begin{tikzpicture}[decoration={brace,amplitude=5pt,mirror}]
		\draw[very thick] (0,0) -- (4,0);
		\draw[semithick] (0,-1) -- (1,-1);
		\draw[semithick] (3,-1) -- (4,-1);
		\draw[semithick] (0,-2) -- (2,-2);
		\draw[semithick] (3,-2) -- (4,-2);
		\draw[semithick] (0,-3) -- (4,-3);
		\draw[semithick] (0,-4) -- (4,-4);
		\draw[very thick] (0,-5) -- (4,-5);
		\draw (1,0) -- (1,-5);
		\draw (2,0) -- (2,-5);
		\draw (3,0) -- (3,-5);
		
		\draw[gray,dashed] (1,-1) -- (3,-1);
		\draw[gray,dashed] (2,-2) -- (3,-2);
		\node at (1,0.5) {$t-1$};
		\node at (2,0.5) {$t$};
		\node at (3,0.5) {$t+1$};
		\node at (1.5,-1.2) {$S'$};
		\node at (2.5,-1.2) {$J'$};
		\node at (2.5,-4.5) {$k$};
		\node at (1.5,-3.5) {$i$};
		\draw[white,decorate,semithick] (1,-5) -- (2,-5) node[midway, yshift =-0.5cm] {$S_i$};
		\draw[white,decorate,semithick] (2,-5) -- (3,-5) node[midway, yshift =-0.5cm] {$J_i \cup \{i\}$};
		
		\draw[->, thick] (1.5,-0.5) -- (2.5,-0.5);
		\draw[->, thick] (1.5,-1.5) -- (2.5,-1.5);
		\draw[->, thick] (1.5,-1.5) -- (2.5,-2.5);
		
		\draw[very thick] (1,0) -- (2,0) -- (2,-2) -- (1,-2) -- (1,0);
		\draw[very thick] (2,0) -- (3,0) -- (3,-3) -- (2,-3) -- (2,0);
		\draw[pattern=north west lines] (2,-3) -- (3,-3) -- (3,-4) -- (2,-4) -- (2,-3);
\end{tikzpicture}
\caption{{\small Relevant time slots $[t-1;t+1]$ in the initial schedule (left) and final schedule (right), respectively. The arrows indicate that the respective job in $S'$ is the predecessor of the corresponding job in $J'$. In this example, $J = \emptyset$, so $J_i = J'$.}}
\label{fig:makespan:preemption:swap}
\end{figure}
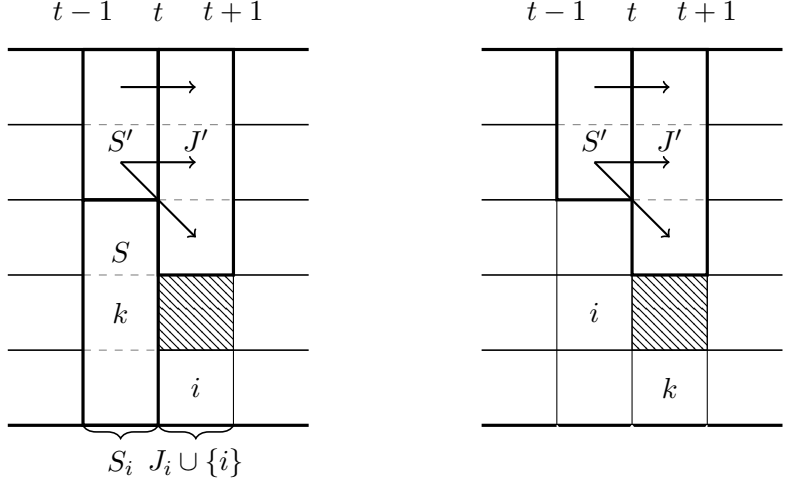

Let $J' = \{j \in J_i \, | \, \mathcal{P}_{L_j}(j) \cap S_i \not= \emptyset\} \cup \{j \in J_i \, | \, \abs{\mathcal{P}(j) \cap \mathcal{A}^*_j} = 1 \text{ and } \mathcal{P}(j) \cap \mathcal{A}^*_j \subseteq S_i\}$ be the set of jobs that are scheduled parallel to $i$ and that are processed directly after their predecessor in the minimal chain or the only predecessor preceding them in the schedule.
Let $S' \subseteq S_i$ be the set of these predecessors of jobs in $J'$.
We do not want to swap $i$ with a job in $S'$ since this would cause an inversion or yield an infeasible schedule.
Note that $\abs{S'} \leq \abs{J'}$, and set $S = S_i \setminus S'$ and $J = J_i \setminus J'$. 
Then $\abs{S} = m - \abs{S'} \geq m - \abs{J'} \geq \abs{J_i} + 1 - \abs{J'}=  \abs{J} +1 \geq 1$, so $S \not= \emptyset$.
We show that any $k \in S$ satisfies the claim.
Figure~\ref{fig:makespan:preemption:swap} illustrates the sets and the corresponding schedules before and after the swap.

Let $k \in S$ be arbitrary and consider the resulting schedule after swapping $k$ and $i$.
Feasibility of the initial schedule implies that at most $m$ jobs are running at any point in time.
It remains to be shown that the precedence constraints are satisfied and that no additional inversions are created.
Let $\mathcal{A}'_j := \{j' \in N^{(p)} \ | \ C'_{j'} < C'_j \}$ for all $j \in N^{(p)}$ be the set of jobs that complete before $j$ starts in the new schedule.

As for feasibility, recall that the precedence constraints for $i$ are not violated if we schedule $i$ in $[t-1;t]$ by assumption.
In the following, let $j \in N^{(p)} \setminus \{i\}$ be a job with predecessors, i.e., $\mathcal{P}(j) \not= \emptyset$.
Note that $\mathcal{P}(j) \cap \mathcal{A}^*_j \not= \emptyset$, since the initial schedule was feasible.
If $j \in N^{(p)} \setminus (J_i \cup \{i\})$, we get $\mathcal{P}(j) \cap \mathcal{A}'_j \not= \emptyset$ because $\mathcal{A}^*_j \subseteq \mathcal{A}'_j$. (Note that strict inclusion only holds for $j = k$.)
For $j \in J'$, its predecessor in $S'$ is still contained in $\mathcal{A}'_j$ since $k \notin S'$.
So $\mathcal{P}(j) \cap \mathcal{A}'_j \not= \emptyset$.
Finally, any job $j \in J_i \setminus J'$ has a predecessor $j' \in \mathcal{P}(j)$ that completes before time $t-1$ by definition of $J'$, so $j' \in \mathcal{A}'_j$.
In total, each job with predecessors is still preceded by one of them, and the schedule is feasible. 

As for the number of inversions, note that the schedule is not altered in the intervals $[0;t-1] \cup [t+1;C^*_{\max}]$.
All jobs that could cause an inversion are contained in $\{i\} \cup S_i$.
Scheduling $i$ one time slot earlier does not cause an inversion by assumption.
The only jobs in $S_i$ that could cause an inversion, if we schedule them in $[t;t+1]$, are contained in $S'$.
Since we swap $i$ with a job in $S = S_i \setminus S'$, swapping $i$ and $k \in S$ does not cause an additional inverted pair.
Hence $I_{C'} \leq I_{C^*}$.
\end{proof}

\begin{lemma}\label{lem:cmax:preemption:structureofOPT}
Let $\{L_j \, | \, j \in N^{(p)}\}$ be a closed collection of minimal chains $L_j \in \mathcal{MC}(j)$ for all $j \in N^{(p)}$.
There exists an optimal preemptive schedule $C^*$ such that $C^*_i < C^*_j$ for all $j \in N^{(p)}$ and $i \in L_j \setminus \{j\}$.
\end{lemma}

\begin{proof}
Recall that all processing times of jobs in $N^{(p)}$ are equal to~$1$.
Consider an optimal schedule with completion times $C^*_j$ for all $j \in N^{(p)}$ such that $I_{C^*}$ is minimal among all optimal solutions.
Suppose by contradiction that $I_{C^*} \geq 1$.
We show how to construct a schedule with $C'_{\max} = C^*_{\max}$ and $I_{C'} < I_{C^*}$ using Lemma~\ref{lem:cmax:swap}.

Since the schedule is optimal, we can assume that the initial job $j^{\text{in}}$ starts at time~0.
Let $j \in N^{(p)}$ and $i \in \mathcal{P}_{L_j}(j)$ such that $(i,j)$ is an inverted pair, i.e., $C^*_i \geq C^*_j \geq mc(j) \geq mc(i) +1$.
We can enumerate the jobs in $L_j = \{j_0,j_1,\dots,j_\ell,j_{\ell +1}\} \in \mathcal{MC}(j)$ such that $j^{\text{in}} = j_0$, $i = j_\ell$, $j = j_{\ell+1}$ and $j_{q-1} \in \mathcal{P}_{L_j}(j_ q)$ for all $q \in [\ell+1]$.
Note that $mc(j_{q-1}) + 1 \leq mc(j_q) \leq C^*_{j_q}$ for all $q \in [\ell+1]$ and $mc(j_0) = mc(j^{\text{in}}) = 0$.

Using Lemma~\ref{lem:cmax:swap}, we move the jobs $j_1,\dots,j_\ell$ successively (in this order) to the front such that they complete at times $mc(j_1),\dots,mc(j_\ell)$, respectively.
For all $k \in N^{(p)}$ with non-trivial release date, it holds
\begin{align}\label{lem:makespan:preemption:noinversion:releasedate}
C^*_k \geq mc(k) \geq r_k + p_k = r_k + 1.
\end{align}
So we can swap those jobs $k \in \{j_1,\dots,j_\ell\}$ that do not complete at time $mc(k)$ to the front without violating the respective release dates.
Thereby, we obtain a schedule that satisfies
\begin{align}\label{lem:makespan:preemption:noinversion:newschedule}
0 = C'_{j_0} < mc(j_1) = C'_{j_1} < mc(j_2) = C'_{j_2} < \cdots < mc(j_\ell) = C'_{j_\ell} < C'_j.
\end{align}
Since we first move job $j_1$ to the front, then $j_2$, and so on, we ensure that, when we apply Lemma~\ref{lem:cmax:swap} for $i = j_q$ (in the notation of Lemma~\ref{lem:cmax:swap}), then its predecessor $j_{q-1}$ completes at time $mc(j_{q-1}) < mc(j_q)$.
So the assumptions of Lemma~\ref{lem:cmax:swap} are satisfied.
The procedure of Lemma~\ref{lem:cmax:swap} does not violate any release dates, since $k \in S$ (in the notation of Lemma~\ref{lem:cmax:swap}) is scheduled later and it is feasible to schedule $j_q$ earlier due to~(\ref{lem:makespan:preemption:noinversion:releasedate}) for all $q \in [\ell]$.

Figure~\ref{fig:makespan:preemption:movetofront} illustrates the current completion times and the time slots in which we move the jobs in the minimal chain $L_j$.
Note that it is not necessary to move the job $j = j_{\ell+1}$.
However, by applying Lemma~\ref{lem:cmax:swap}, it might happen that $k = j$ (in the notation of Lemma~\ref{lem:cmax:swap}) is chosen, i.e., $j$ is ``passively moved''.
Similarly, a job $j_h$ might be ``passively moved'' when we swap $j_q$ with $q < h$ to the front.
This is not a problem, since we deal with $j_h$ in a later iteration.

\begin{figure}
\centering
\begin{tikzpicture}[semithick]
	\tikzmath{\s = 1.2;} 
	\draw[fill=gray!30] (2*\s,1*\s) rectangle (3*\s,0*\s) node[midway] {$j_1$};
	\draw[fill=gray!30] (1*\s,2*\s) rectangle (2*\s,1*\s) node[midway] {$j_2$};
	\draw[fill=gray!30] (7*\s,3*\s) rectangle (8*\s,2*\s) node[midway] {$j_3$};
	\draw[fill=gray!30] (6*\s,1*\s) rectangle (7*\s,0*\s) node[midway] {$j_4$};
	\draw[fill=gray!30] (6*\s,2*\s) rectangle (7*\s,1*\s) node[midway] {$j$};
	
	\draw[arc] (2.2*\s,0.5*\s) -- (0.5*\s,0.5*\s);
	\draw[arc] (7.2*\s,2.5*\s) -- (3.5*\s,2.5*\s);
	\draw[arc] (6.2*\s,0.5*\s) -- (4.5*\s,0.5*\s);
	
	\draw[very thick] (8.5*\s,3*\s) -- (0*\s,3*\s) -- (0*\s,0*\s) -- (8.5*\s,0*\s);
	\foreach \i in {1,2}{ \draw (0*\s,\i*\s) -- (8.5*\s,\i*\s);}
	\foreach \i in {1,2,3,4,5,6,7,8}{ \draw (\i*\s,3*\s) -- (\i*\s,0*\s);}
	
	\draw[dashed] (0,3*\s) -- (0,3.1*\s) node[above] {$0$};
	\draw[dashed] (1*\s,3*\s) -- (1*\s,3.1*\s) node[above] {$mc(j_1)$};
	\draw[dashed] (2*\s,3*\s) -- (2*\s,3.1*\s) node[above] {$mc(j_2)$};
	\draw[dashed] (3*\s,3*\s) -- (3*\s,3.1*\s) node[above] {$r_{j_3}$};
	\draw[dashed] (4*\s,3*\s) -- (4*\s,3.1*\s) node[above] {$mc(j_3)$};
	\draw[dashed] (5*\s,3*\s) -- (5*\s,3.1*\s) node[above] {$mc(j_4)$};

	\draw[arc] (0,-0.5) -- (3,-0.5) node[midway,yshift=-7pt] {time};
\end{tikzpicture}
\caption{{\small Illustration of the procedure to move jobs in $L_j \setminus \{j\} = \{j_1,j_2,j_3,j_4\}$ to the front. Blank squares are jobs not in $L_j$. Arrows indicate into which time slot we want to move the respective jobs. The jobs are moved ``lowest index first'' rather than all at once. Note that $mc(j_3) > mc(j_2) +1$ because $r_{j_3} = 3$.}}
\label{fig:makespan:preemption:movetofront}
\end{figure}
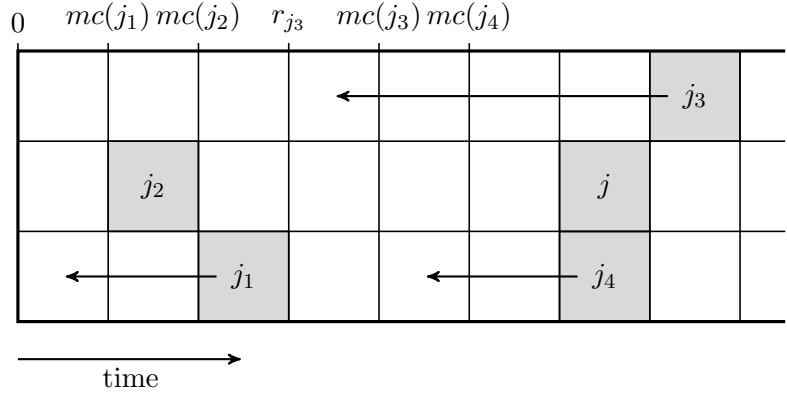

Multiple application of Lemma~\ref{lem:cmax:swap} ensures that the resulting schedule is feasible and has no more inversions than the initial schedule.
Further, Lemma~\ref{lem:cmax:swap} implies $C'_{\max} = C^*_{\max}$, and $I_{C'} < I_{C^*}$ because $i$ and $j$ are not inverted anymore, see~(\ref{lem:makespan:preemption:noinversion:newschedule}).
This contradicts to the choice of the initial schedule being an optimal solution with fewest inversions.
So there exists an optimal solution without inversions, which proves the claim.
\end{proof}

The following lemma shows correctness of \textsc{AlgoPmtn}, and thus proves Theorem~\ref{thm:preemption}.
\begin{lemma}\label{lem:correctness}
\textsc{AlgoPmtn} solves $P \, | \, r_j, \, or\text{-}prec, \, pmtn \, | \, C_{\max}$ to optimality in polynomial time.
\end{lemma}

\begin{proof}
First, observe that the graph $G'$ constructed by \textsc{AlgoPmtn} is a subgraph of the initial precedence graph $G$.
Since the schedule returned by the algorithm is feasible for the AND-instance on $G'$ (this follows from correctness of Lawler's algorithm~\cite{Lawler1982}), it certainly is feasible for the OR-instance on $G$.
Construction of the earliest start schedule and Lawler's algorithm run in polynomial time~\cite{ErlebachKaabMohring2003,Lawler1982}.
Also, we can compute the closed collection of minimal chains and construct $G'$ in polynomial time. 
So \textsc{AlgoPmtn} runs in polynomial time and returns a feasible schedule.

As for optimality of the schedule returned by \textsc{AlgoPmtn}, let $\{L_j \, | \, j \in N \}$ be the closed collection of minimal chains that is computed in the second step, and let $G'$ be the corresponding subgraph of $G$.
Since $\{L_j \, | \, j \in N \}$ is closed, $G'$ is an outforest.
Thus, OR- and AND-precedence constraints on $G'$ are equivalent.

Consider the schedule returned by \textsc{AlgoPmtn}, i.e., by Lawler's algorithm~\cite{Lawler1982} on~$G'$, and let $C_{\max}$ be its makespan.
Since the schedule is feasible for the OR-instance with precedence graph $G'$, it is also feasible for the initial precedence graph $G$.
By Lemma~\ref{lem:cmax:preemption:structureofOPT}, there exists an optimal solution with makespan $C^*_{\max}$ for the instance on $G$ that is also feasible for the instance on $G'$.
Since the schedule returned by \textsc{AlgoPmtn} is optimal for the instance on $G'$, it holds $C_{\max} \leq C^*_{\max}$. This proves the claim.
\end{proof}

\section{Concluding Remarks}\label{sec:remarks}

In this paper, we discuss the problem of minimizing the makespan on parallel uniform machines with OR-precedence constraints.
We introduce the concept of minimal chains, which is crucial to prove that the List Scheduling algorithm of Graham~\cite{Graham1966} achieves an approximation guarantee of $2$.
Using minimal chains, we show that there exists an optimal preemptive schedule of a certain structure and exploit this structure to obtain a polynomial-time algorithm for the preemptive variant.

This matches the complexity and best-known approximation guarantees of makespan minimization, if the precedence graph is an outtree, which is a special case where AND- and OR-precedence constraints coincide.
Clearly any improvement on OR-precedence constraints directly transfers to AND-precedence constraints on outtrees.
On the other hand, due to the close connection with minimal chains, any progress on the approximation factor of AND-precedence constraints on outtrees might also be applicable to OR-precedence constraints.

We would like to remark that Corollary~\ref{cor:unitprocessingtime} (unit processing times) without release dates was already proven by Johannes~\cite{Johannes2005}.
However, the size of the preemptive instance is not polynomial in the input parameters of the initial instance.
Thus the analysis in~\cite{Johannes2005} cannot be extended to the preemptive case.


\printbibliography

\end{document}